\documentclass[11pt,letterpaper]{amsart}

\usepackage%[pagebackref=true]
{hyperref}
\hypersetup{
  colorlinks,
  bookmarksopen,
  bookmarksnumbered,
  citecolor=blue,
  linkcolor=black,
  pdfstartview=FitH,
}

\usepackage{amsmath}
\usepackage{amsfonts}
\usepackage{amssymb}
\usepackage{amsthm}

\newtheorem{thm}{Theorem}[section]
\newtheorem{theorem}[thm]{Theorem}%[section]
\newtheorem{lemma}[thm]{Lemma}

\newtheorem{proposition}[thm]{Proposition}

\theoremstyle{definition}
\newtheorem{remark}[thm]{Remark}

\newcommand{\bb}[1]{\mathbb{#1}}
\newcommand{\cl}[1]{\mathcal{#1}}

\def\dc#1{\expandafter\def\csname#1\endcsname{\mathcal{#1}}}
\def\db#1{\expandafter\def\csname b#1\endcsname{\mathbb{#1}}}
\def\loopy#1#2{%
  \def#1##1{\def\next{#2{##1}#1}\ifx##1\relax\let\next\relax\fi\next}}
\loopy{\makemathcals}{\dc}\loopy{\makemathbbs}{\db}
\makemathbbs  CRNQZTDFPE\relax
\makemathcals QWERTYUIOPASDFGHJKLZXCVBNM\relax

\let\epsilon\varepsilon

\DeclareMathOperator{\dist}{dist}
\newcommand{\dcb}{\dist_{cb}}
\DeclareMathOperator{\id}{id}

\DeclareMathOperator{\conv}{conv}
\DeclareMathOperator{\tr}{tr}
\DeclareMathOperator{\Tr}{Tr}
\DeclareMathOperator{\MU}{MU}

\begin{document}

\title{Schur multipliers and mixed unitary maps}

\author[S.~J.~Harris]{Samuel J.~Harris\textsuperscript{1}}
%\address{Institute for Quantum Computing and Department of Pure Mathematics\\University of Waterloo\\Waterloo\\ON\\Canada\\N2L 3G1}
\email{sj2harri@uwaterloo.ca}
\author[R.~H.~Levene]{Rupert H. Levene\textsuperscript{2}}
%\address{School of Mathematics and Statistics\\University College Dublin\\Belfield\\Dublin 4\\Ireland}
\email{rupert.levene@ucd.ie}
\author[V.~I.~Paulsen]{Vern I.~Paulsen\textsuperscript{1}}
%\address{Institute for Quantum Computing and Department of Pure Mathematics\\University of Waterloo\\Waterloo\\ON\\Canada\\N2L 3G1}
\email{vpaulsen@uwaterloo.ca}
\author[S.~Plosker]{Sarah Plosker\textsuperscript{3}}
%\address{Department of Mathematics and Computer Science\\Brandon University\\Brandon\\MB\\Canada\\R7A 6A9}
\email{ploskers@brandonu.ca}
\author[M.~Rahaman]{Mizanur Rahaman\textsuperscript{1}}
%\address{Institute for Quantum Computing and Department of Pure Mathematics\\University of Waterloo\\Waterloo\\ON\\Canada\\N2L 3G1}
\email{mizanur.rahaman@uwaterloo.ca}

\thanks{\textsuperscript{1}Institute for Quantum Computing and Dept.~of Pure Math., University of Waterloo, Waterloo, Ontario N2L 3G1, Canada}
\thanks{\textsuperscript{2}School of Mathematics and Statistics, University College Dublin, Belfield, Dublin 4, Ireland}
\thanks{\textsuperscript{3}Department of Mathematics and Computer Science, Brandon University, Brandon, MB R7A 6A9, Canada}

\begin{abstract}
We consider the tensor product of the completely depolarising channel on $d\times d$ matrices with the map of Schur multiplication by a  $k \times k$ correlation matrix and characterise, via matrix theory methods, when such a map is a mixed (random) unitary channel. When $d=1$, this recovers a result of O'Meara and Pereira, and for larger $d$ is equivalent to a result of Haagerup and Musat that was originally obtained via the theory of factorisation through von Neumann algebras. We obtain a bound on the distance between a given correlation matrix for which this tensor product is nearly mixed unitary and a correlation matrix for which such a map is exactly mixed unitary. This bound allows us to give an elementary proof of another result of Haagerup and Musat about the closure of such correlation matrices without appealing to the theory of von Neumann algebras. 
\end{abstract}

\maketitle

\section{Introduction}

For any $k\in \bN$, let $M_k$ denote the set of all $k\times k$ complex-valued matrices. A \emph{quantum channel} is a completely positive, trace-preserving linear map $\phi:M_m\rightarrow M_n$; such a map can be written in its (non-unique) Choi-Kraus decomposition as $\phi(X)=\sum_{i=1}^MA_iXA_i^*$, where the $A_i$ are $n\times m$ matrices known as the \emph{Kraus operators}, and ${A}^*$ is the complex conjugate transpose of $A$. Trace-preservation yields $\sum_{i=1}^MA_i^*A_i=I_m$, where $I_m$ denotes the unit of the algebra $M_m$, i.e., the diagonal matrix with 1's on the diagonal. The map $\phi$ is \emph{unital} if $\phi(I_m) = I_n$, which is equivalent to having $\sum_{i=1}^MA_iA_i^*=I_n$. 

Given $d\in \bN$, the \emph{completely depolarising channel} $\delta_d: M_d \to M_d$ is the unital quantum channel given by  $\delta_d(X) = \tr_d(X) I_d$, where we let $\tr_d(X) = \frac{1}{d} \Tr(X)$ denote the normalised trace on $M_d$ with $\tr_d(I_d) =1$.

Let $\U(d)$ be the set of unitary elements of $M_d$. A map $\phi: M_d \to M_d$ is said to be \emph{mixed unitary}
if it is in the convex hull of the maps of the form $X\mapsto UXU^*$ for $U\in \U(d)$. Mixed unitary maps are unital quantum channels, and are sometimes also referred to as \emph{random unitary channels} in the literature, in particular in the context of quantum cryptography.  

Given $C=(c_{i,j})$ and $X=(x_{i,j})$ in $M_k$, their \emph{Schur product} is the matrix $C \circ X = (c_{i,j}x_{i,j})$.
For any $C\in M_k$, the corresponding \emph{Schur multiplier} is the map $S_C:M_k\rightarrow M_k$ given by Schur multiplication: $X\mapsto C\circ X$. The Schur multiplier $S_C$ is a unital quantum channel if and only if $C$ is a \emph{correlation matrix}, i.e., a positive semidefinite matrix whose diagonal elements are all equal to $1$.

O'Meara and Pereira~\cite{per} characterised the set of correlation matrices $C$ such that $S_C$ is  a mixed unitary map, via matrix theory methods.  The work of Haagerup and Musat~\cite{hm11,hm15} on factorisation of completely positive maps through von Neumann algebras can be seen as yielding
higher order versions of the O'Meara-Pereira result, as well as an important asymptotic version that relates limits of certain correlation matrices to Connes' embedding problem.  The relation to Connes' embedding problem uses the work of Dykema and Juschenko~\cite{DJ} on matrices of unitary moments.  Our work is also motivated by the recent work of Musat and R\o rdam~\cite{mr}, which builds on the earlier work of Haagerup and Musat and makes this asymptotic relation to correlation matrices explicit. 

The results of Haagerup and Musat rely on the theory of factorisation through von Neumann algebras and properties of ultrapowers of the hyperfinite II${}_1$ factor in their proofs.
The main goal of this paper is to obtain these higher order analogues of the O'Meara-Pereira result and their asymptotic versions by matrix theoretic methods, without reference to factorisation through von Neumann algebras.  Our main new result, which allows us to take this approach, is a theorem that says, roughly, if $C$ is a correlation matrix such that $\delta_d \otimes S_C$ is \emph{nearly} a mixed unitary map, then $C$ is \emph{near} to a correlation matrix $\widetilde{C}$ such that $\delta_{2d} \otimes S_{\widetilde{C}}$ is a mixed unitary map, with explicit bounds.

In Section~2, for each $d$ and $k$, we characterise the $k \times k$ correlation matrices such that $\delta_d \otimes S_C$ is a mixed unitary map. The result we obtain is equivalent to a corresponding result of~\cite{hm11} once one realizes that \emph{exact factorisation through $M_d \otimes M_k \otimes L^{\infty}[0,1]$} reduces to our statement about convex combinations. In Section~3, we prove the ``nearly'' result and use it to give a new proof of the asymptotic results for correlation matrices of~\cite{hm11,hm15}.

\section{Schur multipliers with a finite mixed-unitary ancilla}\label{Sec:genP}

In this section we obtain ``higher order'' versions in the spirit of the result of~\cite{per}.

The following result is likely well known, but we prove it here for completeness. Given a subset $S$ of a vector space, we use $\conv(S)$ to denote the, not necessarily closed, set of convex combinations of elements of $S$. 

\begin{lemma}\label{lem:compactconv}
  Let $A_d$ be a subset of the closed unit ball of $\bR^m$ for $d\in\bN$. If $\overline{\bigcup_d A_d}$ is convex, then \[\overline{\bigcup_d A_d}=\overline{\bigcup_{d}\conv(A_d)}.\]
\end{lemma}
\begin{proof}
  Let $A=\overline{\bigcup_d A_d}$ and $B=\overline{\bigcup_{d}\conv(A_d)}$. Plainly, $A\subseteq B$. Conversely, let $b\in B$. Then $b=\lim_{n\to \infty}  b_n$ for some $b_n\in\conv(A_{d_n})$ where $d_n\in \bN$. By Carath\'{e}odory's theorem, there is $N\in \bN$ (independent of $n$) so that  $b_n=\sum_{l=1}^N p_{n,l}a_{n,l}$ for some probability distribution $p_{n,1},\dots,p_{n,N}\in [0,1]$ and some $a_{n,l}\in A_{d_n}$. By compactness of $[0,1]$ and of $A$, we may pass to a subsequence for which  $(p_{n,l})_n$ and $(a_{n,l})_n$ are convergent, say to $p_l\in [0,1]$ and $a_l\in A$. So \[b=\lim_{n\to \infty} \sum_{l=1}^N p_{n,l}a_{n,l}=\sum_{l=1}^N p_la_l\in \conv(A)=A, \]
as desired. \end{proof}

We now define some interesting sets of correlation matrices. For any $d,k\in \bb{N}$, let \[\cl{F}_{k}(d)=\{(\tr_d(U_i^*U_j))_{i,j=1}^k\in M_k\,:\, U_1, \dots, U_k\in \U(d)\}.\] 
 Since the normalised trace $\tr_d\colon M_d\to \bC$ is unital and completely positive, it follows that $\F_k(d)$ is indeed a set of correlation matrices.  We remark that
since the convex hull of a compact set in $\bR^n$ is compact, the sets $\F_k(d)$ and $\conv(\cl{F}_k(d))$ are both compact, for any $d\in \bN$. 
%We then define
%\[ \cl{FS}_k = \overline{\bigcup_{d\in \bb{N}}\cl{FS}_{k}(d)}.\]
Adopting the notation of~\cite{DJ}, we define \[\cl{F}_k=\overline{\bigcup_{d\in \bb{N}}\cl{F}_{k}(d)}.\]
By Lemma~\ref{lem:compactconv}, we also have
\[ \cl{F}_k=\overline{\bigcup_{d\in \bb{N}}\conv(\cl{F}_{k}(d))}.\]

We use $E_{i,j}= |i\rangle \langle j |$ to denote the standard matrix units. 
Let $d\in \bN$ and $\omega=\exp\left(\frac{2\pi i}{d}\right)$. %The Weyl
%operators in $\U(d)$ generalize the Pauli matrices defined for
%$d=2$. They are the unitary group~$\W$ of order $d^2$ generated by
%$U, V\in \U(d)$,
Let $S=\sum_{j=1}^d E_{j+1,j} \in M_d$ be the
forward cyclic shift operator and let $D$ be the diagonal operator given by
$D=\sum_{j=1}^d \omega^jE_{j,j} \in M_d$.  The Weyl-Heisenberg unitaries are the $d^2$ unitaries given by
\[ W_{a,b} = S^a D^b, \, \, \, 0 \le a,b \le d-1.\]
This set is a projective unitary group, i.e., modulo scalars it is a group.

Recall that the completely depolarising channel on $M_d$ is the map
\[\delta_d\colon M_d\to M_d,\quad \delta_d(X)=\tr_d(X)I_d.\]
It is not difficult to show~(see~\cite[Chapter~4]{watrous}) that
\begin{equation}
\label{eq:delta-MU}
\delta_d(X) = \frac{1}{d^2} \sum_{a,b=0}^{d-1} W_{a,b}XW_{a,b}^*,\end{equation} 
so $\delta_d$ is mixed unitary.

We recall a theorem of O'Meara and Pereira:
\begin{theorem}[{\cite[Theorem~5]{per}}]\label{thm:per} Let $k\in \bN$ and let $C\in M_k$ be a correlation
  matrix. The Schur multiplier $S_C\colon M_k\to M_k$ is mixed unitary
  if and only if $C$ is in $\conv (\F_k(1))$, i.e., $C$ lies in the convex hull of the rank-one correlation matrices in $M_k$.
\end{theorem}
We generalise Theorem~\ref{thm:per} as follows:
\begin{theorem}\label{thm:gen}
  Let $C\in M_k$ be a correlation matrix and let $d\in \bN$. The map $\delta_d\otimes S_C\colon M_d\otimes M_k\to M_d\otimes M_k$ is mixed unitary if and only if 
  \[ C\in \conv(\cl{F}_k(d)).
  %\conv\{ (\tr_d (U_iU_j^*))_{i,j}\in M_k\colon U_i\in \U(d),\;1\le i\le k\}.
  \]
\end{theorem}
\begin{proof}
  Suppose  $S_C\colon M_k\to M_k$ is such that $\delta_d\otimes S_C$ is mixed unitary. Then for $A=(A_{i,j})=\sum_{i,j=1}^k A_{i,j}\otimes E_{i,j}\in M_d\otimes M_k$, we have
  \[
    \delta_d\otimes S_C(A) = (c_{i,j}\tr_d(A_{i,j})I_d) = \sum_{l=1}^M p_l V_l A V_l^*
  \]
  for some unitaries $V_l\in \U(dk)$, some probability distribution $p_1,\dots,p_M>0$ and some $M\in \bN$.
  For $A=I_d\otimes E_{i,i}$, we have
  \[ \delta_d\otimes S_C(I_d\otimes E_{i,i})=I_d\otimes E_{i,i} = \sum_{l=1}^M p_l V_l(I_d\otimes E_{i,i}) V_l^*.\]
  Write $V_l=(V_{l,i,j})_{i,j=1}^k=\sum_{i,j=1}^k V_{l,i,j}\otimes E_{i,j}\in M_d\otimes M_k$ where each $V_{l,i,j}$ is in $M_d$. The right hand side becomes
  \[\sum_{l=1}^M p_l\sum_{s,t=1}^k V_{l,s,i}V_{l,t,i}^*\otimes E_{s,t}.\]
  For $s=t\ne i$, we obtain
  \[ 0=\sum_{l=1}^M p_l V_{l,s,i}V_{l,s,i}^*\] and each term in the latter sum is positive semidefinite. Hence $V_{l,s,i}=0$
  whenever $s\ne i$, and $V_{l}=\bigoplus_{i=1}^k V_{l,i,i}$. In
  particular, since $V_{l} \in \U(dk)$, each $V_{l,i,i}$ is in $\U(d)$. On the other hand, taking
  $A=I_d\otimes E_{i,j}$ where $i\ne j$, we obtain
  \begin{eqnarray*}
    \delta_d\otimes S_C(I_d\otimes E_{i,j})&=&c_{i,j} I_d\otimes E_{i,j} = \sum_{l=1}^M p_l\sum_{s,t=1}^k \left(V_{l,s,i}V_{l,t,j}^*\otimes E_{s,t}\right)\\
    &=& \left(\sum_{l=1}^M p_lV_{l,i,i}V_{l,j,j}^*\right)\otimes E_{i,j},
  \end{eqnarray*}
  hence
  \[ c_{i,j}=\tr_d(c_{i,j} I_d) = \sum_{l=1}^M p_l \tr_d
    (V_{l,i,i}V_{l,j,j}^*)\] and
  \[ C=(c_{i,j}) = \sum_{l=1}^M p_l \left( \tr_d(U_{l,i}U_{l,j}^*)\right)\]
  where $U_{l,i}=V_{l,i,i}\in \U(d)$. Therefore,  $ C\in \conv({\F}_k(d))$, % \conv\{ (\tr_d (U_iU_j^*))_{i,j}\colon U_i\in \U(d),\;1\le i\le k\}
  as claimed.

  \newcommand{\Wt}{\widetilde W}
  Consider the converse. Since the set of mixed unitary maps is convex, so is the set of maps $\Phi:M_k \to M_k$ such that $\delta_d \otimes \Phi$ is mixed unitary, for a fixed~$d$. Thus, it suffices to establish the converse in the case where $C=(c_{i,j})$ is in $\F_k(d)$. Then $C$ is of the form $C=(\tr_d (U_iU_j^*))_{i,j=1}^k $ for some $U_i\in \U(d)$, where $\;1\le i\le k$. Let $\{W_l\}_{l=1}^{d^2}$ be an enumeration of the Weyl operators in $\U(d)$. Define $\Wt_{l,l'} = \bigoplus_{i=1}^k W_{l'}U_iW_{l}\in \U(dk)$. By Equation~(\ref{eq:delta-MU}),
%  we have
% \[ \delta_d(X) = d^{-2}\sum_{l} W_l X W_l^*.\]
%   Then
  for $A=(A_{i,j})\in M_{k}(M_d)$ we have
  \begin{eqnarray*}
    d^{-4}\sum_{l,l'=1}^{d^2} \Wt_{l,l'} (A_{i,j}) \Wt_{l,l'}^* &=&d^{-4}\sum_{l,l'=1}^{d^2}(W_{l'} U_i W_l A_{i,j} W_{l}^* U_j^* W_{l'}^*)\\
                                                &=&d^{-2}\sum_{l'=1}^{d^2} ( W_{l'} U_i \tr_d(A_{i,j})U_j^* W_{l'}^*)\\
                                                &=&(\tr_d(A_{i,j})\tr_d(U_iU_j^*)I_d)\\
                                                &=&(\tr_d(A_{i,j})c_{i,j}I_d)\\
                                                &=& (\delta_d\otimes S_C)(A)
  \end{eqnarray*}
  which shows that the map $\delta_d\otimes S_C$ is indeed mixed unitary. 
\end{proof}
\begin{remark}
We recover Theorem~\ref{thm:per} by taking $d=1$.
\end{remark}
 
\begin{remark} In~\cite{hm15}, a quantum channel $T:M_k \to M_k$ is called \emph{factorisable of degree d} if and only if $\delta_d \otimes T$ is a mixed unitary. In~\cite[Proposition 3.4]{hm15}, they prove that a channel is factorisable of degree~$d$ if and only if it has an \emph{exact factorisation through $M_d \otimes M_k \otimes L^{\infty}[0,1]$.} By unraveling what this latter property means in the case of a Schur product map and using the fact that the the set of mixed unitary channels is closed, one obtains Theorem~\ref{thm:gen}.
\end{remark}

\begin{remark} Musat and R\o rdam~\cite{mr} prove the remarkable result that  $\bigcup_d \conv(\cl F_k(d))$ is not closed for any $k \ge 11$.
\end{remark}

%%%%%%%%%%%%%%%%%%%%
%%%%%%%%%%%%%%%%%
%%%%%%%%%%%%
%%%
\section{Asymptotically mixed unitary Schur multipliers}
We write $\MU(dk)$ or $\MU(M_d\otimes M_k)$ for the set of mixed
unitary maps on $M_{dk}=M_d\otimes M_k$. %Plainly, any mixed unitary map is a
%unital quantum channel. 
Note that mixed unitary maps are closed under
several natural operations, including taking convex combinations,
tensor products and composition. Given an arbitrary map $\Phi: M_{dk} \to M_{dk}$ we set
\[\dist(\Phi, \MU(dk)) = \inf \{ \| \Phi - \Psi \| : \Psi \in \MU(dk) \}\]and\[\dcb( \Phi, \MU(dk) ) = \inf \{ \| \Phi - \Psi \|_{cb} : \Psi \in \MU(dk) \}\] where $\| \cdot \|_{cb}$ denotes the completely bounded norm~\cite{Pa} of a map. For $d$ and $k$ fixed these distances are comparable, since for any map $\Psi$ on $M_n$ one has~\cite{Pa} $\|\Psi \| \le \|\psi \|_{cb} \le n \|\Psi \|$. 
One goal of this section is to give a matrix theoretic proof of a result of~\cite{hm15} that shows that $C$ is in the closure of $ \bigcup_d \conv( \cl F_k(d))$ if and only if $\dcb( \delta_d \otimes S_C, \MU(dk)) \to 0$ as $d\rightarrow \infty$. In fact, we will prove a somewhat stronger result, namely that  $C$ is in the closure of $ \bigcup_d \conv( \cl F_k(d))$ if and only if $\dist( \delta_d \otimes S_C, \MU(dk)) \to 0$ as $d\rightarrow \infty$.

Since $\delta_d\colon M_d\to M_d$ is in $\MU(M_d)$, the idempotent map
$\Delta_{d,k}:=\delta_d\otimes \id_{M_k}$ is in $\MU(M_d\otimes M_k)$,
for any $d,k\in \bN$. 

We now consider the effect of ``$\Delta$-compression''.

\begin{lemma}\label{lem:MU} Let $d,k\in \bN$ and 
 let $\Phi:M_d\otimes M_k\rightarrow M_d\otimes M_k$ be a unital quantum channel, and write $\Delta=\Delta_{d,k}$.
Then $\Delta\circ\Phi\circ\Delta=\delta_d\otimes T$ for some unital quantum channel  $T:M_k\rightarrow M_k$.  Moreover, if $\Phi\in \MU(dk)$, then $\Delta\circ \Phi\circ \Delta\in \MU(dk)$.
\end{lemma}

\begin{proof}
  Let $\Psi=\Delta\circ \Phi\circ \Delta$, which is a unital quantum channel on
  $M_d\otimes M_k$. Every element of the range of $\Delta$ is of the form $I_d\otimes Y$ for some $Y\in M_k$. Since $\Psi=\Delta\circ \Psi$, it follows that there exist linear maps $\Psi_i\colon M_k\to M_k$ for $1\le i\le d$ so that
  \[ \Psi(E_{i,i}\otimes B)=I_d\otimes \Psi_i(B),\quad B\in M_k.\]
  Define
  \[ T=\sum_{i=1}^d \Psi_i.\]
  Then for $A\in M_d$ and $B\in M_k$, since $\Psi=\Psi\circ \Delta$, we have
  \begin{eqnarray*}
    \Psi(A\otimes B)&=& \Psi(\Delta(A\otimes B))\\
                    &=& \tr_d(A)\Psi(I_d\otimes B)\\
                    &=& \tr_d(A)\sum_{i=1}^d \Psi(E_{i,i}\otimes B)\\
                    &=& \tr_d(A)\sum_{i=1}^dI_d\otimes \Psi_i(B)\\
                    &=& \tr_d(A)I_d\otimes \sum_{i=1}^d\Psi_i(B)\\
                    &=& \delta_d(A)\otimes T(B),
  \end{eqnarray*}
  so $\Psi=\delta_d\otimes T$. This implies that $T(B)=(\tr_d\otimes \id_{M_k})\circ \Psi(I_d\otimes B)$; since the maps $\tr_d\otimes \id_{M_k}$ and $\Psi$ are both unital quantum channels, as is the embedding $B\mapsto I_d\otimes B$, so is~$T$.
  
 The final assertion follows immediately, since the set of mixed unitary maps is closed under composition.
\end{proof}
Given linear maps $\phi: M_n \to M_n$ and $\psi: M_k \to M_k$ we have a map $\psi \otimes \phi: M_k \otimes M_n \to M_k \otimes M_n$, and $\| \psi \otimes \phi \| \ge \|\psi \| \cdot \|\phi\|$. Generally, this inequality can be strict. For example, taking $\phi$ to be the identity map and computing the supremum over all $n$ is how one obtains the completely bounded norm of $\psi$.
However, for completely bounded maps we have that~\cite{Pa}  \[\|\phi \otimes \psi \|_{cb} = \|\phi \|_{cb} \cdot \|\psi \|_{cb}.\] For a unital completely positive map $\phi$, we have that $\|\phi\|_{cb} = 1$.  From these facts, it follows that when $\phi$ is a unital completely positive map and $\psi$ is completely bounded, we have  $\|\phi \otimes \psi \|_{cb} = \|\psi\|_{cb}$.

Our first result shows that the above inequality becomes an equality for the completely depolarizing chanel $\delta_d$.

\begin{proposition}\label{prop:depoleq} Let $d,k \in \bN$ and let $R: M_k \to M_k$ be a linear map. Then
$\| \delta_d \otimes R\|= \|R\|$.
\end{proposition}
\begin{proof} Let $X_{i,j} \in M_k, \, 1 \le i,j \le d$ with $\| \sum_{i,j=1}^d E_{i,j} \otimes X_{i,j} \|=1$ in $M_d \otimes M_k$. We have that
  \begin{align*}
    \| (\delta_d \otimes R)(\sum_{i,j=1}^d E_{i,j} \otimes X_{i,j} ) \| &= \| \frac{1}{d} \sum_{i=1}^d I_d \otimes R(X_{i,i}) \| \\&\le \max \{ \|R(X_{i,i}) \| : 1 \le i \le d \} \le \|R\|,
  \end{align*}
  since $\|X_{i,i}\| \le 1$, and the result follows.
\end{proof}

\begin{proposition}\label{prop:dist to MU} Let $d \in \bb N$ and let $R: M_k \to M_k$ be a quantum channel.  Then there are quantum channels $T_i: M_k \to M_k$, $i=1,2$ such that $\delta_d \otimes T_i$, $i=1,2$ are mixed unitary with
\[ \|R - T_1 \|_{cb} = \dcb( \delta_d \otimes R, \MU(dk)),\]
and 
\[ \|R - T_2 \| = \dist( \delta_d \otimes R, \MU(dk)).\]
\end{proposition}
\begin{proof} Since the set of mixed unitaries is a compact set, we may choose a mixed unitary $\Phi_1: M_d \otimes M_k \to M_d \otimes M_k$ such that 
\[\|\delta_d \otimes R - \Phi_1 \|_{cb} = \dcb( \delta_d \otimes R, \MU(dk)).\]
By the above,  $\Delta \circ \Phi_1 \circ \Delta = \delta_d \otimes T_1$ for some unital quantum channel $T_1 \colon M_k\to M_k$.
Then
\begin{eqnarray*}
\|R - T_1 \|_{cb} &=& \|\delta_d \otimes R - \delta_d \otimes T_1 \|_{cb} = \| \Delta \circ [ \delta_d \otimes R - \Phi_1] \circ \Delta \|_{cb} \\ &\le & \| \delta_d \otimes R - \Phi_1 \|_{cb} = \dcb( \delta_d \otimes R, \MU(dk)).\end{eqnarray*}
However,
\[ \|R - T_1 \|_{cb} = \|\delta_d \otimes R - \delta_d \otimes T_1\|_{cb} \ge \dcb( \delta_d \otimes R, \MU(dk)),\]
and so the first result follows.

The proof of the second result is similar. One first picks $\Phi_2$ that attains the norm distance and precedes as in the first case while using Proposition~\ref{prop:depoleq}. 
\end{proof}
Our main theorem is an analogue of Proposition~\ref{prop:dist to MU} in the case that $R$ is a Schur product map, except that we would also like to choose $T$ to be a Schur product map. In this case we do not get an equality, but we are able to get a bound.

Working towards our main theorem requires a certain averaging by unitary conjugation. Typically, averaging by unitary conjugation is called \emph{twirling}; in the quantum information theory literature, it is said that one is applying a \emph{twirling operation}. The twirl of a quantum channel is again a quantum channel. Here we do a somewhat different operation, which preserves complete positivity, but does not generally preserve the property of being a quantum channel; this operation was previously used in~\cite[Section~V]{ydx}. We use $\{e_j\}$ to denote the standard basis  of vectors of a given dimension. 

\newcommand{\dint}{\int\!\!\int}
\begin{lemma}\label{lem:Daverage}
  Let $T\colon M_k\to M_k$ be completely positive and let $\gamma\colon M_k\to M_k$ be the ``$\D$-biaverage'' of $T$, given by
  \[ \gamma(X)=\dint
    D_1^*T(D_1XD_2)D_2^*\, dD_1\,dD_2,\] where the integrals are taken
  over the group of diagonal unitary matrices in $M_k$ with respect to Lebesgue
  measure on $\bT^k$. Then $\gamma=S_B$ for the positive semidefinite matrix $B=V^*C_TV$ where $C_T$ is the Choi matrix of $T$ given by  $C_T=(T(E_{i,j}))_{i,j=1}^k\in M_{k}\otimes M_k$ and $V:\bC^k\rightarrow \bC^k\otimes \bC^k$ is the isometry $e_j\mapsto e_j\otimes e_j$. In particular, $\gamma$ is completely positive.
\end{lemma}

\begin{proof}
  We have
  \begin{eqnarray*}
    \gamma(E_{i,j})&=& \dint D_1^*T(D_1 E_{i,j} D_2)D_2^*\,dD_1\,dD_2\\
                   &=& \dint D_1^*d_{1,i} T( E_{i,j} )d_{2,j}D_2^*\,dD_1\,dD_2\\
                   &=& \left(\int D_1^*d_{1,i}\,dD_1\right) T( E_{i,j} ) \left(\int d_{2,j}D_2^*\,dD_2\right)\\
                   &=& E_{i,i}T( E_{i,j} ) E_{j,j}\\
                   &=& T(E_{i,j})_{i,j} E_{i,j}.
  \end{eqnarray*}
  So $\gamma$ is Schur multiplication by the matrix $B=(b_{i,j})$,
  where $b_{i,j}=T(E_{i,j})_{i,j}$. Since $T$ is completely positive,
  its Choi matrix $C_T=(T(E_{i,j}))_{i,j=1}^k\in M_{k}\otimes M_k$ is
  positive, and $B$ is the compression of $C_T$ to the subspace
  spanned by $\{e_i\otimes e_i\colon 1\le i\le k\}$. Hence, $B=V^*C_TV$ is
  positive semidefinite. Consequently, $\gamma=S_B$ is completely positive.
\end{proof}
\begin{remark} The map $\gamma$ of Lemma~\ref{lem:Daverage} can also be obtained by averaging over the $2^k$ diagonal matrices of $\pm 1$'s.
\end{remark}

\begin{remark} The same proof shows that, if  $\cl G \subseteq \cl U(k)$ is any compact subgroup, $\phi: M_k \to M_k$ is completely positive, and we set
\[ \gamma(X) = \int \int U_1^*\phi(U_1XU_2)U_2^* dU_1 dU_2, \]
where $dU$ denotes Haar measure on $\cl G$, then $\gamma$ is completely positive and \emph{$\cl G$-covariant}, i.e., $\gamma(U_1XU_2) = U_1 \gamma(X)U_2$ for any $U_1, U_2 \in \cl G$.   We shall refer to $\gamma$ as the \emph{$\cl G$-biaverage} of $\phi$. Note that the $\cl G$-biaverage of a quantum channel need not be a quantum channel, but it can be shown to be trace non-increasing for positive elements. Similarly, if $\phi$ is unital, then $0 \le \gamma(I_k) \le I_k$.  Both of these latter inequalities follow by showing that if $\{ A_i \}$ is a set of Choi-Kraus operators for $\phi$, then $\{ \bb E(A_i) \}$ is a set of Choi-Kraus operators for $\gamma$, where $\bb E:M_k \to M_k$ is the conditional expectation onto the commutant of $\cl G$, and using the Cauchy-Schwarz inequality for completely positive maps.
\end{remark}

In finite dimensions, we have the following unitary
dilation at our disposal. 
\begin{lemma}\label{lem:udilation}
  If $X\in M_d$ with $\|X\|\le 1$, then there exist $A,B \in M_d$ such that the $(2d) \times (2d)$ matrix
  \[
    W=\left[\begin{array}{cc}
      X&A\\B&X
    \end{array}\right]
  \] is unitary.
\end{lemma}
\begin{proof}
  Define $C=\sqrt{I_d-XX^*}$ and $D=\sqrt{I_d-X^*X}$, and let $X=UP$ be the polar decomposition of~$X$, where, since we work in finite dimensions, we may assume that $U$ is unitary (rather than merely a partial isometry) and $P$ is positive. Halmos' unitary dilation of~$X$ is given by
  \[\left[\begin{array}{cc}
   X&C\\D&-X^*  
  \end{array}\right]\in \U(2d).\]
Thus,
\[ W:=
  \left[\begin{array}{cc}
    I_d&0\\0&-I_d
   \end{array}\right]
  \left[\begin{array}{cc}
    I_d&0\\0&U
 \end{array}\right]
  \left[\begin{array}{cc}
    X&C\\D&-X^*  
   \end{array}\right]
 \left[\begin{array}{cc}
    I_d&0\\0&U
   \end{array}\right]\in \U(2d)\]
since $W$ is a product of unitaries. Since $UX^*U=UPU^*U=UP=X$, we see that $W$ has the desired form.
\end{proof}
\begin{remark} In infinite dimensions, taking $X$ to be the unilateral shift, one sees that there are no operators $A$ and $B$ such that the above operator matrix is a unitary.
\end{remark}

%%%%%%%%%%%%%%%%%%%%%%%%%%%%%%%%%%%%%%%%%%%%%%%%%%
% MAIN THEOREM
%%%%%%%%%%%%%%%%%%%%%%%%%%%%%%%%%%%%%%%%%%%%%%%%%%

\begin{theorem}\label{thm:main}
 Let $C=(c_{i,j})$ be a $k\times k$ correlation matrix and let $\epsilon>0$.
 If $\dist(\delta_d\otimes S_C,\MU(dk))<\epsilon$, then there is $\widehat C=(\widehat c_{i,j})\in \conv(\F_k(2d))$ with $\|S_C-S_{\widehat{C}}\| <2\epsilon$. In particular, $|c_{i,j}-\widehat c_{i,j}|<2\epsilon$ for all $1 \leq i,j \leq k$.
\end{theorem}

\begin{proof}
  By hypothesis, there is a mixed unitary $\Phi:M_{dk}\rightarrow M_{dk}$ with
  \[\|\delta_d\otimes S_C-\Phi\| <\epsilon.\] By Lemma~\ref{lem:MU},
  the mixed unitary map $\Delta\circ \Phi\circ \Delta$ is of the form
  $\delta_d\otimes T$ for some unital quantum channel $T\colon M_k\to M_k$. Since
  $\Delta\circ (\delta_d\otimes S_C)\circ \Delta=\delta_d\otimes S_C$,
  by Proposition~\ref{prop:depoleq} we have
  \[\|S_C-T\|= \|\delta_d\otimes S_C-\delta_d\otimes T\|=\|\Delta\circ
    (\delta_d\otimes S_C-\Phi)\circ \Delta\| <\epsilon.\]
  Since $\delta_d\otimes T$ is mixed unitary, we may write
  \[ \delta_d\otimes T(X)=\sum_{l=1}^M t_l U_lXU_l^*,\quad X\in M_{dk},\] for some $M\in \bN$ and
  unitaries $U_1,\dots,U_M\in \U(dk)$ and $t_l\ge 0$ satisfying
  $\sum_{l=1}^M t_l=1$. By Lemma~\ref{lem:Daverage}, the $\cl D$-biaverage of $T$ is a positive Schur multiplier, say $S_{\widetilde C}\colon M_k\to M_k$, for some positive semidefinite $\widetilde C\in M_k$. For any $X\in M_k$ with $\|X\|\le 1$, we have
\begin{align*}  \|S_C(X)-S_{\widetilde C}(X)\|&=\left\|\dint
    D_1^*S_C(D_1XD_2)D_2^*-D_1^*T(D_1XD_2)D_2^*\, dD_1\,dD_2\right\|\\
                    &\le\dint\left\| D_1^*\big(S_C(D_1XD_2)-T(D_1XD_2)\big)D_2^*\right\|\, dD_1\,dD_2\\
  &\le \dint \| S_C-T\|\,dD_1\,dD_2\le \|S_C-T\| <\epsilon.
\end{align*}
Hence, $\|S_C-S_{\widetilde C}\|< \epsilon$. This implies that $|c_{i,j}-\widetilde c_{i,j}|<\epsilon$
for every $i,j$. In particular, since $c_{i,i}=1$ for all~$i$, we have
\begin{equation}
  \label{eq:ctildebound}
  |1-\widetilde c_{i,i}|<\epsilon,\quad 1\le i\le k.
\end{equation}

For $X\in M_{dk}$, we have
  \begin{align*}
   \delta_d &\otimes S_{\widetilde C}(X)=\sum_l t_l \dint (I_d\otimes D_1)^* U_l(I_d\otimes D_1)X(I_d\otimes D_2)U_l^* (I_d\otimes D_2)^*\,dD_1\,dD_2\\
                            & =\sum_l t_l \left(\int (I_d\otimes D_1)^* U_l(I_d\otimes D_1)\,dD_1\right)
                              X
                              \left(\int(I_d\otimes D_2)U_l^* (I_d\otimes D_2)^*\,dD_2\right)\\
                            & = \sum_l t_l X_l X X_l^*
  \end{align*}  
  where $X_l=\int (I_d\otimes D)^* U_l(I_d\otimes D)\,dD\in M_{dk}$. Since
  $X_l$ lies in the convex hull of $\U(dk)$, we have $\|X_l\|\le 1$.

To see that each $X_l$ is block diagonal, just as in the proof of Theorem~\ref{thm:gen}, we calculate
\[ \delta_d\otimes S_{\widetilde C}(I_d\otimes E_{i,i})=I_d\otimes \widetilde{c}_{i,i}E_{i,i} = \sum_{l=1}^M t_l X_l(I_d\otimes E_{i,i}) X_l^*.\]
  Write $X_l=(X_{l,i,j})_{i,j=1}^k=\sum_{i,j=1}^k X_{l,i,j}\otimes E_{i,j} \in M_d \otimes M_k$; then we obtain
  \[I_d\otimes \widetilde{c}_{i,i}E_{i,i}=\sum_{l=1}^M t_l\sum_{s,t=1}^k X_{l,s,i}X_{l,t,i}^*\otimes E_{s,t}.\]
  For $s=t\ne i$, we have
  \[ 0=\sum_{l=1}^M t_l X_{l,s,i}X_{l,s,i}^*.\] By positivity, $X_{l,s,i}=0$
  whenever $s\ne i$. Writing $X_{l,i}:=X_{l,i,i}$, it follows that $X_{l}=\bigoplus_{i=1}^k X_{l,i}=\sum_{i=1}^kX_{l,i}\otimes E_{i,i}$. Since $\|X_l\| \le 1$, it follows that $\|X_{l,i}\|\le 1$ for each $i$. Moreover,
  \[ I_d\otimes \widetilde c_{i,j}E_{i,j}=\delta_d\otimes S_{\widetilde C}(I_d\otimes E_{i,j}) = \sum_{l=1}^M t_l X_l(I_d\otimes E_{i,j}) X_l^*=  \sum_{l=1}^M t_l X_{l,i} X_{l,j}^*\otimes E_{i,j},\]
  so
  \[ \widetilde c_{i,j}E_{i,j}=\tr_d \otimes \id(I_d \otimes \widetilde{c}_{i,j} E_{i,j})=\tr_d\otimes \id\left(\sum_{l=1}^M t_lX_{l,i}X_{l,j}^*\otimes E_{i,j}\right).\] Hence,
  \begin{equation}
    \label{eq:ctilde}
    \widetilde c_{i,j}=\sum_{l=1}^M t_l\tr_d(X_{l,i}X_{l,j}^*).
  \end{equation}
  
  Applying Lemma~\ref{lem:udilation} to each $X_{l,i}$, we obtain unitary matrices $W_{l,i}\in \U(2d)$ of the form
  \[ W_{l,i}=
  \left[\begin{array}{cc}
      X_{l,i}&A_{l,i}\\
      B_{l,i}&X_{l,i}
    \end{array}\right],
  \]

for some $A_{l,i},B_{l,i} \in M_d$. Now, consider
    \[ \widehat C:=(\widehat c_{i,j})\in \conv (\F_k(2d))\]
    defined by
    \begin{eqnarray*}
      \nonumber\widehat c_{i,j}
      &:=&\sum_{l=1}^Mt_l\tr_{2d}(W_{l,i}W_{l,j}^*)
      \\\nonumber      
      &=&
        \sum_{l=1}^Mt_l\tr_{2d}\left[\begin{array}{cc}
          X_{l,i}X_{l,j}^*+A_{l,i}A_{l,j}^*&*\\
          *&B_{l,i}B_{l,j}^*+X_{l,i}X_{l,j}^*
        \end{array}\right]
      \\\nonumber
      &=&\sum_{l=1}^Mt_l\tr_{d}( X_{l,i}X_{l,j}^*)+\frac12\left(\sum_{l=1}^Mt_l\left(\tr_{d}\left(A_{l,i}A_{l,j}^*+B_{l,i}B_{l,j}^*\right)\right) \right)
      \\
      &=& \widetilde c_{i,j}+\frac12\left(\sum_{l=1}^Mt_l\left(\tr_{d}\left(A_{l,i}A_{l,j}^*+B_{l,i}B_{l,j}^*\right)\right) \right)%\label{eq:conv}
    \end{eqnarray*}
    where the off-diagonal terms denoted by $*$ in the second line may be ignored, as they do not affect the trace. 
    
    Since each $W_{l,i}$ is unitary, we have $X_{l,i}X_{l,i}^*+A_{l,i}A_{l,i}^*=I_d$. By Equation~(\ref{eq:ctilde})  above, we have
    \begin{eqnarray*}
      \sum_{l=1}^Mt_l\tr_d(A_{l,i}A_{l,i}^*)
      &=&\sum_{l=1}^Mt_l\tr_d(I_d-X_{l,i}X_{l,i}^*)\\
      &=&1-\sum_{l=1}^Mt_l\tr_d(X_{l,i}X_{l,i}^*)\\
      &=&1-\tr_{d}(\widetilde{c}_{i,i}I_d)\\
      &=&1-\widetilde{c}_{i,i}.
    \end{eqnarray*}
    In particular, by Equation~(\ref{eq:ctildebound}),
    \[\left| \sum_{l=1}^M t_l \tr_d(A_{l,i}A_{l,i}^*) \right|<\epsilon.\]
    Define $y_{i,j}=\sum_{l=1}^M t_{l} \tr_d(A_{l,i}A_{l,j}^*)$ and $z_{i,j}=\sum_{l=1}^M t_{l} \tr_d(B_{l,i}B_{l,j}^*)$, and set $Y=(y_{i,j})_{i,j=1}^k$ and $Z=(z_{i,j})_{i,j=1}^k$. Then $|y_{i,i}|<\epsilon$ for each $i$. A similar argument shows that $|z_{i,i}|<\epsilon$. We have $\widehat{C}=\widetilde{C}+\frac{1}{2} (Y+Z)$, so
    \[S_{\widehat{C}}-S_{\widetilde{C}}=\frac{1}{2} (S_Y+S_Z).\]
    We will show that $\|S_Y\| <\epsilon$; the argument for $S_Z$ is similar.  For each $1 \leq l \leq M$, the matrix $(A_{l,i} A_{l,j}^*)_{i,j=1}^k \in M_k \otimes M_d$ is positive. Then $(\tr_d(A_{l,i}A_{l,j}^*))=\id_k \otimes \tr_d (A_{l,i}A_{l,j}^*)_{i,j}$ is positive as well. Taking convex combinations, we see that $Y$ is positive in $M_k$. In particular, $0 \leq y_{i,i}<\epsilon$ for each $i$. But since $Y$ is positive, the Schur multiplier map $S_Y$ is completely positive, so that \[\|S_Y\| =\|S_Y(I_k)\|=\max \{ y_{i,i}: 1 \leq i \leq k \}<\epsilon.\]
    Similarly, $\|S_Z\| <\epsilon$, so that
    \[\| S_{\widehat{C}}-S_{\widetilde{C}} \| \leq \frac{1}{2} (\|S_Y\|+\|S_Z\|)<\frac{1}{2} (\epsilon+\epsilon)=\epsilon.\]
    Finally, since $\|S_C - S_{\widetilde{C}}\| <\epsilon$, it follows that
    \[\|S_C-S_{\widehat{C}}\| \leq \|S_C-S_{\widetilde{C}}\| +\|S_{\widetilde{C}}-S_{\widehat{C}}\| <2\epsilon.\]
    Hence,
    \begin{eqnarray*}
    |c_{i,j}-\widehat{c}_{i,j}| = \|(S_C-S_{\widehat{C}})(E_{i,j})\| \leq \|S_C-S_{\widehat{C}}\| <2\epsilon 
    \end{eqnarray*}
    as desired.
\end{proof}
\begin{theorem}\label{thm:AB}
  Let $C\in M_k$ be a correlation matrix. The following are equivalent:
  \begin{enumerate}
  \item $C\in \F_k$;
  \item  $\inf_{d} \dcb(\delta_d\otimes S_C,\MU(dk))=0$;
  \item $\inf_d \dist( \delta_d \otimes S_C, \MU(dk)) =0$.
\end{enumerate}
\end{theorem}
\begin{proof}
  Given $C\in \F_k$, there exist $d_n\in \bN$ and $C_n\in \F_k(d_n)$ with $C_n\to C$, entrywise, as $n\to \infty$. Then \[\dcb(\delta_{d_n}\otimes S_C,\MU(d_nk))\le \|\delta_{d_n}\otimes S_C-\delta_{d_n}\otimes S_{C_n}\|_{cb} = \|S_C-S_{C_n}\|_{cb}\to 0\]
  as $n\to \infty$, since all norms are equivalent in finite dimensions. Hence the infimum above is $0$.
  
  Clearly, the second statement implies the third, since \[\dist( \delta_d \otimes S_C, \MU(dk)) \le \dcb ( \delta_d \otimes S_C, \MU(dk)).\]
  
  Finally, suppose the infimum in (3) is $0$ for $C=(c_{i,j})$. By Theorem~\ref{thm:main}, given $\epsilon>0$, there exists $d_\epsilon\in\bN$ and $\widehat{C}=(\widehat{c}_{i,j})\in \conv(\F_k(2d_\epsilon))$ such that $|c_{i,j} - \widehat{c}_{i,j}| <2\epsilon$, for all $1 \le i,j \le k$.  It follows that \[C\in \overline{\bigcup_{\epsilon>0} \conv(\F_k(2d_\epsilon))}\subseteq \F_k,\] which completes the proof. \end{proof}
 
 \begin{remark} The equivalence of (2) and (3) in the above theorem is somewhat surprising, since crude estimates give
 \[ \dist (\delta_d \otimes S_C, \MU(dk)) \le \dcb ( \delta_d \otimes S_C, \MU(dk)) \le dk \cdot \dist( \delta_d \otimes S_C, \MU(dk)).\]
 We do not know if better bounds, that are independent of $d$ and $k$, can be obtained for these two distances to the mixed unitaries. It is possible that, for a mixed unitary map $\Phi$, we have that $\| \delta_d \otimes S_C-\Phi\|_{cb} \leq M \| \delta_d \otimes S_C - \Phi\|$ for some constant $M>0$ that is independent of $d$ and $k$.
\end{remark}

\begin{remark} By~\cite[(3.15)]{hm11}, \[\dcb( \delta_{d+1} \otimes T, \MU((d+1)k)) \le \dcb(\delta_d \otimes T, \MU(dk)) + \frac{1}{d+1},\]so that
\[ \inf_d \dcb(\delta_d \otimes T, \MU(dk)) = \lim_{d\rightarrow\infty} \dcb( \delta_d \otimes T, \MU(dk)).\]
A similar estimate shows that 
\[ \inf_d \dist(\delta_d \otimes T, \MU(dk)) = \lim_d \dist (\delta_d \otimes T, \MU(dk)).\]
 Our last theorem should be compared to~\cite[Theorem~3.6]{hm15}, which also proves the equivalence of three statements. Their first two statements are our (1) and (2) (with the infimum replaced by the limit) and their third statement involves factorisation through an ultrapower of the hyperfinite II${}_1$ factor. Their proof that (2) implies (1) first shows that (2) implies this factorisation result, then that this factorisation result implies (1).
\end{remark}

\section*{Acknowledgements}
The authors would like to thank the referee who suggested that the results of Section 3 should be true without reliance on the cb-norm.

RHL is grateful to the University College Dublin Seed Funding Visiting Professors programme
for their support. VIP is supported by the Natural Sciences and Engineering Research Council (NSERC) grant number 03784. SP is supported by NSERC Discovery Grant number 1174582, the Canada Foundation for Innovation (CFI) grant number 35711, and the Canada Research Chairs (CRC) Program grant number 231250. RHL and SP wish to acknowledge the Institute for Quantum Computing,
University of Waterloo for their kind hospitality during their visits in
June 2018. MR holds a
Postdoctoral Fellowship in Pure Mathematics at the University of
Waterloo.

\bigskip
\end{document}